\documentclass{article}

\usepackage{amsmath,amsthm,amssymb}
\usepackage{url}
\usepackage{subfigure}
\usepackage{graphicx}
\usepackage{amsmath,amssymb}
\usepackage{amstext}
\usepackage{dsfont}
\usepackage{xspace}
\usepackage{epsfig}
\usepackage{balance}
\usepackage[noend]{algorithmic}
\usepackage{algorithm}
\usepackage{color}

\newtheorem{definition}{Definition}
\newtheorem{theorem}{Theorem}
\newtheorem{lemma}{Lemma}
\newtheorem{corollary}{Corollary}

\renewcommand{\epsilon}{\varepsilon}
\newcommand{\R}{\mathds{R}}
\newcommand{\N}{\mathds{N}}

\newcommand{\ceil}[1]{\left\lceil#1\right\rceil}
\newcommand{\E}[1]{\mathrm{E}\left(#1\right)}

\newcommand{\Prob}[1]{\mathrm{Pr}\left[#1\right]}

\newcommand{\LO}{\textup{LO}\xspace}

\newcommand{\OM}{\textup{OneMax}\xspace}
\newcommand{\ONEMAX}{\OM}
\newcommand{\Jump}{\textup{Jump}\xspace}

\newcommand{\EA}{\text{(1+1)~EA}\xspace}
\newcommand{\EAL}{\text{(1+$\lambda$)~EA}\xspace}

\newcommand{\ie}{i.\,e.\xspace}

\newcommand{\migint}{\tau}

\newcommand{\paralleltime}{T^{\mathrm{par}}}
\newcommand{\sequentialtime}{T^{\mathrm{seq}}}

\newcommand{\positive}[1]{\left(#1\right)^+}
\newcommand{\algo}{\ensuremath{\mathcal{A}}}

\newcommand{\ignore}[1]{\ensuremath{}}

\usepackage[plainpages=false,pdfpagelabels]{hyperref}

\begin{document}
\title{Adaptive Population Models for Offspring Populations and~Parallel Evolutionary Algorithms}

\author{J{\"o}rg L{\"a}ssig\\
ICS, University of Lugano\\
6906 Lugano, Switzerland
\and Dirk Sudholt\\
CERCIA, University of Birmingham\\
 Birmingham B15 2TT, UK
}

\maketitle
\begin{abstract}
We present two adaptive schemes for dynamically choosing the number of parallel instances in parallel evolutionary algorithms. This includes the choice of the offspring population size in a (1+$\lambda$)~EA as a special case.
Our schemes are parameterless and they work in a black-box setting where no knowledge on the problem
is available.
Both schemes double the number of instances in case a generation ends without finding an improvement. In a successful generation, the first scheme resets the system to one instance, while the second scheme halves the number of instances. Both schemes provide near-optimal speed-ups in terms of the parallel time. We give upper bounds for the asymptotic sequential time (i.\,e., the total number of function evaluations) that are not larger than upper bounds for a corresponding non-parallel algorithm derived by the fitness-level method.
\end{abstract}

\section{Introduction}
\label{sec:Introduction}

Parallelization is becoming a more and more important issue for solving difficult optimization problems~\cite{Alba2005}. Various implementations of parallel evolutionary algorithms (EAs) have been applied in the past decades
\cite{Tomassini2005}.
An obvious way of using parallelization is to parallelize single operations of an EA such as executing fitness evaluations on different processors. This particularly applies to EAs using large offspring populations.
So-called \emph{island models} use parallelization on a higher level. The idea is to parallelize evolution itself, by having subpopulations, called islands, which evolve in parallel. Good solutions are exchanged between the islands in a \emph{migration} process.

One of the most important questions when dealing with parallel EAs is how to choose the number of processors in order to decrease the \emph{parallel optimization time}, defined as the number of generations until an EA has found a global optimum.
Assume a setting where we can choose the number of processors to be allocated, but we have to pay costs for each processor in each generation it is being used. This situation is common in cloud computing or in large grids where processors are shared with other users. The total cost for all processors over time is called \emph{sequential optimization time}.
The task is now to choose the number of processors to be used such that the parallel optimization time is small, but at the same time the sequential time is reasonable.
Allocating too many processors would waste computational effort and hence unnecessarily increase the sequential optimization time. Allocating too few processors implies a large parallel optimization time.

During the run of an EA, the ``ideal'' value for the number of processors is likely to change over time. One typical situation is that in the beginning of a run improvements are easy to obtain and only few processors are needed. The better the best fitness, the tougher it gets to find further improvements and then more processors are required. It therefore makes sense to look at adaptive mechanisms that can adjust the number of processors which are being used during the run of the EA.
This obviously only makes sense in a setting where allocating and deallocating processors on-the-fly is possible and the cost for these operations and the cost for the communication between the processors are rather small. Hence we focus on balancing the parallel and sequential optimization times.

In the following we present adaptive schemes for choosing the number of processors that apply both to offspring populations as well as island models of EAs.
We accompany our schemes by a rigorous theoretical analysis of their running time.
Both schemes double the number of processors if the current generation fails to produce an offspring that has larger fitness than the current best fitness value. Otherwise, if the generation yields an improvement, the number of processors is decreased again. The difference between the two schemes lies in the way the number of processors is decreased.

The first scheme, called Scheme~A, simply resets the number of processors to~1; only the best individual or island survives. This is to avoid an overly large number of processors when moving from a situation where improvements are hard to find to a situation where improvements are easy. This happens, for instance, if the EA escapes from a local optimum and then jumps to the basin of attraction of a better local optimum.

The second scheme, Scheme~B, tries to maintain a fair number of processors over time; it also doubles the population size in unsuccessful generations and it halves the population size in successful generations. This strategy makes more sense in case the EA encounters similar probabilities for improvements over time. Both schemes are parameterless and oblivious with respect to the objective function. They can be applied in a black-box setting where no knowledge is available about the problem.

In terms of offspring populations we consider the \EAL
that maintains a single best individual and in each iteration creates $\lambda$ offspring. A best offspring replaces its parent if its fitness is not worse. The $\lambda$ offspring creations and function evaluations can be parallelized on $\lambda$ processors.
Concerning island models, we assume that migration sends copies of each island's best individual to each other island in every generation. So, whenever one island finds an improvement of the current best individual in the system, this is immediately communicated to all other islands. The island model then behaves similarly to offspring populations, but it is more general as the islands can work with populations of size larger than 1.

To unify the notation for island models and offspring populations, we simply speak of the \emph{population size} in the following; this means the number of islands in the island model and the offspring population size for the \EAL, respectively.

For EAs using either Scheme A or B we show that the expected parallel optimization time can be decreased drastically. In comparison to the well-known fitness-level method, in the parallel optimization time for every fitness value the expected waiting time for an improvement can be replaced by its logarithm. This can drastically reduce the parallel optimization time, in particular for problems where improvements are hard to find.
The expected sequential time remains reasonable. We prove upper bounds on the expected sequential optimization time that are asymptotically no larger than upper bounds for a single instance obtained via the fitness-level method.
For problems where the fitness-level method gives tight bounds, our results show that the two schemes automatically yield decreased expected parallel optimization times, without increasing the expected sequential time.

The mentioned bounds are general in the sense that they apply to islands running arbitrary elitist algorithms.
Example applications are given that apply simultaneously to the \EAL and to islands of population size~1. Various functions are considered: \ONEMAX, \LO, the class of unimodal functions and $\Jump_k$.

Comparing the different schemes, our results indicate that Scheme B is more efficient than A, from an asymptotic perspective, as it quickly reduces the number of processors, if necessary. This adaptation automatically leads to optimal or near-optimal parallel optimization times on all considered examples. On one example Scheme~B outperforms Scheme~A. We also compare these results with tailored schemes that are allowed to use knowledge on the objective function.

Besides the main results this paper is also interesting because of the methods used. We introduce new techniques from the amortized analysis of algorithms, which represent natural and effective tools for analyzing adaptive mechanisms. These techniques may find further applications in the analysis of adaptive stochastic search algorithms.

The remainder of this work is structured as follows. In Section~\ref{sec:previous-work} we review previous work. Section~\ref{sec:algorithms} presents the algorithms and the considered population update schemes. In Section~\ref{sec:tail-bounds-and-expectations} we provide technical statements that will be used later on in our analyses and that may also help to understand the dynamics of the adaptive algorithms. Section~\ref{sec:general-upper-bounds} then presents general upper bounds for both schemes, while Section~\ref{sec:lower-bounds} deals with lower bounds on expected sequential times.
Section~\ref{sec:non-oblivious} contains a brief discussion about tailored, that is, non-oblivious population update schemes.
Our general theorems are applied to concrete example functions in Section~\ref{sec:example}. We finish with a discussion of possible extensions in Section~\ref{sec:extensions} and conclusions in Section~\ref{sec:Conclusions}.

\section{Previous Work}
\label{sec:previous-work}

\subsection{Adaptive Population Models}

Considering adaptive numbers of islands in the island model of EAs, previous work is very limited. However, there are numerous results for adaptive population sizes in EAs. Eiben, Marchiori, and Valko~\cite{eiben2004evolutionary} describe EAs with on-the-fly population size adjustment.
They compared the performance of the different strategies in terms of success rate, speed, and solution quality, measured on a variety of fitness landscapes. The best EAs with adaptive population resizing outperformed traditional approaches when considering the time to result, which is the parallel optimization time. Typical approaches are eliminating population size as an explicit parameter by introducing aging and maximum lifetime properties for individuals \cite{michalewicz1996genetic}, the parameter-less GA (PLGA) which evolves a number of populations of different sizes simultaneously \cite{harik1999parameter}, random variation of the population size \cite{costa1999experimental}, and competition schemes \cite{schlierkamp1994strategy}.

Schwefel \cite{schwefel1981numerical} first suggested the adaptation of the offspring population size during the optimization process.
Herdy \cite{herdy1993number} proposed a mutative adaptation of $\lambda$ in a two-level ES, where on the upper level, called population level, $\lambda$ is treated as a variable to be optimized while on the lower level, called individual level, the object parameters are optimized.

In \cite{hansen1995sizing}, a deterministic adaptation scheme for $\lambda$ based on theoretical considerations on the relation between serial rates of progress for the actual number of offspring $\lambda$, for $\lambda - 1$ and for the optimal number of offspring is introduced. More specific, the local serial progress (\ie, progress per fitness function evaluation) is optimized in a $(1,\lambda)$~EA with respect to the number of offspring $\lambda$. The authors prove the following structural property: the serial progress-rate as a function of $\lambda$ is either a function with exact one (local and global) maximum or a strictly monotonically increasing function.

Jansen, De~Jong, and Wegener~\cite{Jansen2005a} further elaborate on the offspring population size, presenting a thorough runtime analysis of the effects of the offspring population size. They also suggest a simple way to dynamically adapt this parameter and present empirical results for this scheme, but no theoretical analysis of their scheme has been performed. The presented scheme doubles the offspring population size if the algorithm is unsuccessful to improve the currently best fitness value. Otherwise, it divides the current offspring population size by $s$, where $s$ is the number of offspring with better fitness than the best fitness value so far.
We will discuss in Section~\ref{sec:extensions} how our schemes relate to their scheme and in how far our results can be transferred.

\subsection{Theoretical Work on Parallel EAs}

In \cite{Lassig2010}, a first rigorous runtime analysis for island models has been performed by constructing a function where alternating phases of independent evolution and communication among the islands are essential. A simple island model with migration finds a global optimum in polynomial time, while panmictic populations as well as island models without migration need exponential time, with very high probability.

New methods for the running time analysis of parallel evolutionary algorithms with spatially structured populations have been presented in \cite{Lassig2010a}. The authors generalized the well known fitness-level method, also called method of $f$\nobreakdash-based partitions~\cite{Wegener2002}, from panmictic populations to spatially structured evolutionary algorithms with various migration topologies. These methods were applied to estimate the speed-up gained by parallelization in pseudo-Boolean optimization.
It was shown that the possible speed-up for the parallel optimization time increases with the density of the topology. The expected sequential optimization time is asymptotically not larger than an upper bound for a corresponding non-parallel EA, derived via the fitness-level method.

More precisely, the classical fitness level method says that when $s_i$ is a lower bound on the probability that one island leaves the current fitness level towards a better one, the expected time until this happens is at most $1/s_i$ for a panmictic population. In a parallel EA with a unidirectional ring, the expected parallel time decreases to $O(s^{1/2})$; in other words, the waiting time can be replaced by its square root. For a torus graph even the third root can be used and with a proper choice of the number $\mu$ of islands, a speed-up of order $\mu$ is possible in some settings.

Interestingly, the results from \cite{Lassig2010a}
can partially be interpreted in terms of adaptive population sizes. The analyses are based on the number of individuals on the current best fitness level. In our upper bounds, we pessimistically assume that only islands on the current best fitness level have a reasonable chance of finding better fitness levels. All worse individuals are ignored when estimating the waiting time for an improvement of the best fitness level. If a unidirectional ring topology is used, migration happens in every generation, and better individuals are guaranteed to win in the selection step, the number of individuals on the current best fitness level increases by 1 in each generation as always a new island is taken over. (We pessimistically ignore the fact that islands on worse fitness levels can improve their best fitness.)
If any island finds an improvement, it is pessimistically assumed that then only one island has made it to a new, better fitness level.
This setting corresponds exactly to a parallel EA that in each unsuccessful generation acquires one new processor and to an adaptive \EAL that increases $\lambda$ by 1 in each unsuccessful generation. Once an improvement is found, the population size drops to~1 as in the case of our first scheme presented here. The upper bounds from~\cite{Lassig2010a} therefore directly transfer to additive population size adjustments. In the following we show that multiplicative adjustments of the population size may admit better speed-ups than additive approaches as suggested in \cite{Lassig2010a}.

\section{Algorithms}
\label{sec:algorithms}

In Sections~\ref{sec:general-upper-bounds} and~\ref{sec:non-oblivious} we present general upper bounds via the fitness-level method.
These results are general in the following sense. If all islands in a parallel EA run elitist algorithms (\ie, algorithms where the best fitness in the population can never decrease) and if we have a lower bound on the probability of finding a better fitness level then this can be turned into an upper bound for the expected sequential and parallel running times of the parallel EA.

We present a scheme for algorithms where this argument applies. The goal is to maximize some fitness function~$f$ in an arbitrary search space. An adaptation towards minimization is trivial.

\begin{algorithm}[H]
    \caption{Elitist parallel EA with adaptive population} \label{alg:parallelEA}
\begin{algorithmic}[1]
\STATE Let $\mu:=1$ and initialize a single island $P_1^1$ uniformly at random.
\FOR{$t:=1$ to $\infty$}
\FORALL{$1\leq i \leq \mu$ in parallel}
\STATE Select parents and create offspring by variation.
\STATE Send a copy of a fittest offspring to all other islands.
\STATE Create $P_{t+1}^i$ such that it contains a best individual from the union of $P_t^i$, the new offspring, and the incoming migrants.
\STATE $\mu_{t+1}:=\text{updatePopulationSize}(P_t^i, P_{i+1}^i)$
\STATE \textbf{if} $\mu_{t+1} > \mu_t$ \textbf{then} create $\mu_{t+1}-\mu_t$ new islands by copying existing islands.
\STATE \textbf{if} $\mu_{t+1} < \mu_t$ \textbf{then} delete $\mu_{t}-\mu_{t+1}$ islands.
\ENDFOR
\ENDFOR
\end{algorithmic}
\end{algorithm}
The selection of islands to be copied or removed, respectively, is left unspecified. Note that each island migrates individuals to all other islands. This corresponds to a complete migration topology. Due to this fact, all islands always contain an offspring with the current best fitness. This observation is sufficient for the upcoming analyses. With other topologies this selection would be based on the fitness values of the current elitists on all islands.

The \EAL can be regarded a special case where we have $\lambda$ islands and a single best individual takes over all $\lambda$ islands. Setting $\lambda := 1$ yields the well-known \EA.
\begin{algorithm}[H]
    \caption{\EAL with adaptive population} \label{alg:EAL}
\begin{algorithmic}[1]
\STATE Initialize a current search point $x_1$ uniformly at random.
\FOR{$t:=1$ to $\infty$}
\STATE Create $\lambda$ offspring by mutation.
\STATE Let $x^*$ be an offspring with maximal fitness.
\STATE \textbf{if} $f(x^*) \ge f(x_t)$ \textbf{then} $x_{t+1} := x^*$ \textbf{else} $x_{t+1} := x_t$.
\STATE $\lambda:=\text{updatePopulationSize}(\{x_t\}, \{x_{t+1}\})$
\ENDFOR
\end{algorithmic}
\end{algorithm}

Note that we have neither specified a search space nor variation operators.
However, in Section~\ref{sec:lower-bounds} we will discuss lower bounds that only hold in pseudo-Boolean optimization and for EAs that only use standard mutation (\ie, flipping each of $n$ bits independently with probability $1/n$) for creating new offspring.

In Section~\ref{sec:example} we will consider concrete example functions where the \EAL with adaptive populations or, equivalently, an island model running \EA{}s, with an adaptive number of islands are applied. The latter was called parallel \EA in~\cite{Lassig2010,Lassig2010a}.

We now define the population update schemes considered in this work.
The function \text{updatePopulationSize} takes the old and the new population as inputs and it outputs a new population size.

In order to help finding improvements that take a long time to be found, we double the population size in each unsuccessful generation. As we might not need that many islands after a success, we reset the population size to~1.

\begin{algorithm}[H]
    \caption{updatePopulationSize$(P_t,P_{t+1})$ (Scheme A)} \label{alg:AlgorithmA}
\begin{algorithmic}[1]
\IF{$\max\{f(x) \mid x \in P_{t+1}\} \le \max\{f(x) \mid x \in P_t\}$}
\STATE \textbf{return} $2 \mu_t$
\ELSE
\STATE \textbf{return} $1$
\ENDIF
\end{algorithmic}
\end{algorithm}

On problems where finding improvements takes a similar amount of time, it might not make sense to throw away all islands at once. Especially if improvements have similar probabilities over time, it makes sense to stay close to the current number of islands. Therefore, in the following scheme we halve the population size with every successful generation. We will see that this does not worsen the asymptotic performance compared to Scheme~A. For some problems this scheme will turn out to be superior.

\begin{algorithm}[H]
    \caption{updatePopulationSize$(P_t,P_{t+1})$ (Scheme B)} \label{alg:AlgorithmB}
\begin{algorithmic}[1]
\IF{$\max\{f(x) \mid x \in P_{t+1}\} \le \max\{f(x) \mid x \in P_t\}$}
\STATE \textbf{return} $2 \mu_t$
\ELSE
\STATE \textbf{return} $\lfloor \mu_t/2\rfloor$
\ENDIF
\end{algorithmic}
\end{algorithm}

The motivation for considering Scheme A is that we can assess the effect of gradually decreasing the population size, when comparing it to Scheme B. It also serves as a first step towards analyzing Scheme B, where the analysis turns out to be more involved.

Our schemes for parallel EAs are applicable in large clusters where the cost of allocating new processors is low, compared to the computational effort spent within the evolutionary algorithm. Many of our results can be easily adapted towards algorithms that do not use migration and population size updates in every generation, but only every $\migint$ generations, for a parameter $\migint \in \N$, called \emph{migration interval}.
This can significantly reduce the costs for allocating and deallocating new processors.
Details can be found at the end of Section~\ref{sec:general-upper-bounds}.

An algorithm using Scheme~B can be implemented in a decentralized way as follows, where we assume that each island runs on a separate processor. Assume all processors are synchronized, \ie, they share a common timer. All processors have knowledge on the current best fitness level and they inform all other processors by sending messages in case they find a better fitness level. This message contains individuals that can be taken over by other processors so that all processors work on the current best fitness level.

In the adaptive scheme, if after one generation no message has been received, \ie, no processor has found a better fitness level, each processor activates a new processor as follows. Each processor maintains a unique ID. The first processor has an ID that simply consists of an empty bit string. Each time a processor activates a new processor, it copies its current population and its current ID to the new processor. Then it appends a 0-bit to its ID while the new processor appends a 1-bit to its ID. At the end, all processors have enlarged their IDs by a single bit. When an improvement has been found, all processors first take over the genetic material in the messages that are passed. Then all processors whose ID ends with a 1-bit shut down. All other processors remove the last bit from their IDs.

It is easy to see that with this mechanism all processors will always have pairwise distinct IDs and no central control is needed to acquire and shut down processors.

As mentioned in the introduction, we define the \emph{parallel optimization time} $\paralleltime$ as the number of generations until the first global optimum is evaluated. The \emph{sequential optimization time} $\sequentialtime$ is defined as the number of function evaluations until the first global optimum is evaluated. The number of function evaluations is a common performance measure and it captures the total effort on all processors. Note that this includes all function evaluations in the generation of the algorithm in which the improvement is found. These definitions are consistent with the measures as suggested in the literature \cite{Jansen2005a}. In both measures we allow ourselves to neglect the cost of the initialization as this only adds a fixed term to the running times.

\section{Tail Bounds and Expectations}
\label{sec:tail-bounds-and-expectations}

In preparation for upcoming running time analyses we first prove tail bounds for the parallel optimization times in a setting where we are waiting for a specific event to happen. This, along with bounds on the expected parallel and sequential waiting times, will be useful to prove our main theorems. The tail bounds  also indicate that the population will not grow too large.
In the remainder of this paper we abbreviate $\max\{x, 0\}$ by $\positive{x}$.
\begin{lemma}
\label{lem:tail-bounds-and-expectations}
Assume starting with $2^k$ islands for some $k \in \N_0$ and doubling the number of islands in each generation.
Let $\paralleltime(k, p)$ denote the random parallel time until the first island encounters an event that occurs independently on each island and in each generation with probability~$p$. Let $\sequentialtime(k, p)$ be the corresponding sequential time. Then for every $\alpha \in \N_0$
\begin{enumerate}
\item
\mbox{$
\Prob{\paralleltime(k,p) > \positive{\ceil{\log(1/p)} -k}\!\! + \alpha + 1} \le \exp(-2^\alpha),
$}
\item
$
\Prob{\paralleltime(k,p) \le \log(1/p) -k - \alpha} \le 2 \cdot 2^{-\alpha}$,
\item
$
\log(1/p) -k - 3 < \E{\paralleltime(k,p)} < \positive{\log(1/p) -k} + 2$, 
\item
$
\max\{1/p, 2^k\} \le \E{\sequentialtime(k,p)} \le 2/p + 2^k - 1.
$
\end{enumerate}
Each inequality remains valid if $p$ is replaced by a pessimistic estimation of~$p$ (\ie, either an upper bound or a lower bound).
\end{lemma}
\begin{proof}
The condition $\paralleltime(k,p) > \positive{\ceil{\log(1/p)} -k}\! + \alpha + 1$ requires that the event does not happen on any island in this time period. The number of trials in the last generation is at least $2^{\ceil{\log(1/p)} + \alpha} \ge 1/p \cdot 2^{\alpha}$ for all~$k \in \N_0$.
Hence
\begin{align*}
\Prob{\paralleltime(k,p) > \positive{\ceil{\log(1/p)} -k} + \alpha + 1} \le\;& (1-p)^{1/p \cdot 2^{\alpha}}\\
 \le\;& \exp(-2^{\alpha})\;.
\end{align*}

For the second statement we assume $k \le \log(1/p)-\alpha$ as otherwise the claim is trivial.
A necessary condition for $\paralleltime(k,p) \le \log(1/p) -k - \alpha$ is that the event does happen at least once within in the first $\log(1/p) -k - \alpha$ generations. This corresponds to at most $\sum_{i=1}^{\log(1/p)-\alpha} 2^{i-1} \le 2^{\log(1/p) - \alpha} = 1/p \cdot 2^{-\alpha}$ trials.
If $p > 1/2$ the claim is trivial as either the probability bound on the right-hand side is at least 1 or the time bound is negative, hence we assume $p \le 1/2$.
Observing that then $1/p \cdot 2^{-\alpha} \le 2(1/p-1) \cdot 2^{-\alpha}$,
the considered probability is bounded by
\begin{align*}
1-(1-p)^{2(1/p-1) \cdot 2^{-\alpha}} \;\le\;& 1-\exp(-2 \cdot 2^{-\alpha})\\
\le\;& 1-(1-2 \cdot 2^{-\alpha}) = 2 \cdot 2^{-\alpha}\;.
\end{align*}

To bound the expectation we observe that the first statement implies $\Prob{\paralleltime(k,p) \ge \positive{\log(1/p) -k} + \alpha + 2} \le \exp(-2^\alpha)$.
 Since $\paralleltime(k,p)$ is non-negative, we have
\begin{align*}
&\E{\paralleltime(k,p)}
\;=\; \sum_{t=1}^{\infty} \Prob{\paralleltime(k,p) \ge t}\\
\le\;& \positive{\log(1/p) -k} + 1\\
& \quad + \sum_{\alpha=0}^\infty \Prob{\paralleltime(k,p) \ge \positive{\log(1/p) -k} + \alpha + 2}\\
\le\;& \positive{\log(1/p) -k} + 1 + \sum_{\alpha=0}^\infty \exp(-2^\alpha)\\
<\;& \positive{\log(1/p) -k} + 2
\end{align*}
as the last sum is less than~1.
For the lower bound we use that the second statement implies $\Prob{T \ge \log(1/p) -k - \alpha} \ge 1-2 \cdot 2^{-\alpha}$.
Hence
\begin{align*}
&\E{\paralleltime(k,p)}
\;=\; \sum_{t=1}^{\infty} \Prob{\paralleltime(k,p) \ge t}\\
\ge\;& \sum_{\alpha=2}^{\log(1/p)-k-1} \Prob{\paralleltime(k,p) \ge \log(1/p) -k - \alpha}\\
\ge\;& \sum_{\alpha=2}^{\log(1/p)-k-1} (1-2 \cdot 2^{-\alpha})\\
=\;& \log(1/p)-k-2 - \sum_{\alpha=1}^{\log(1/p)-k-2} 2^{-\alpha}\\
>\;& \log(1/p)-k-3\;.
\end{align*}

For the fourth statement consider the islands one-by-one, according to some arbitrary ordering. Let $T(p)$ be the random number of sequential trials until an event with probability~$p$ happens. It is well known that $\E{T(p)} = 1/p$. Obviously $\sequentialtime(k,p) \ge T(p)$ since the sequential time has to account for all islands that are active in one generation. This proves $\E{\sequentialtime(k,p)} \ge \E{T(p)} \ge 1/p$. The second lower bound $2^k$ is obvious as at least one generation is needed for a success.

For the upper bound observe that $\sequentialtime(k,p) = 2^k$ in case $T(p) \le 2^k$ and $\sequentialtime(k,p) = \sum_{i=k}^\ell 2^{i}$ in case $\sum_{i=k}^{\ell-1} 2^{i} < T(p) \le \sum_{i=k}^\ell 2^{i}$. Together, we get that $\sequentialtime(k,p) \le \max\{2T(p), 2^k\} \le 2T(p)+2^k-1$, hence $\E{\sequentialtime(k,p)} \le 2/p + 2^k-1$.
\end{proof}

The presented tail bounds indicate that the population typically does not grow too large. The probability that the number of generations exceeds its expectation by an additive value of~$\alpha+1$ is even an inverse doubly exponential function. The following provides a more handy statement in terms of the population size. It follows immediately from Lemma~\ref{lem:tail-bounds-and-expectations}.
\begin{corollary}
Consider the setting described in Lemma~\ref{lem:tail-bounds-and-expectations}. For every $\beta \ge 1$, $\beta$ a power of 2, the probability that while waiting for the event to happen the population size exceeds $\max\{2^{k+1}, 4/p\} \cdot \beta$ is at most $\exp(-\beta)$.
\end{corollary}

One conclusion from these findings is that our schemes can be applied in practice without risking an overly large blowup of the population size.
We now turn to performance guarantees in terms of expected parallel and sequential running times.

\section{Upper Bounds via Fitness Levels}
\label{sec:general-upper-bounds}

The following results are based on the fitness-level method, also known as method of $f$-based partitions (see, e.\,g., Wegener~\cite{Wegener2002}).
This method is well known for proving upper bounds for algorithms that do not accept worsenings of the population.
Consider a partition of the search space into sets $A_1, \dots, A_m$ where for all $1 \le i \le m-1$ all search points in $A_i$ are strictly worse than all search points in $A_{i+1}$ and $A_m$ contains all global optima.
If each set $A_i$ contains only a single fitness value then the partition is called a \emph{canonic} partition.

If $s_i$ is a lower bound on the probability of creating a search point in $A_{i+1} \cup \dots \cup A_m$, provided the current best search point is in $A_i$, then the expected optimization time is bounded from above by
\[
\sum_{i=1}^{m-1} \Prob{A_i} \cdot \sum_{j=i}^{m-1} \frac{1}{s_j}\;,
\]
where $\Prob{A_i}$ abbreviates the probability that the best search point after initialization is in $A_i$. The reason for this bound is that the expected time until $A_i$ is left towards a higher fitness level is at most $1/s_i$ and each fitness level, starting from the initial one, has to be left at most once.
Note that we can always simplify the above bound by pessimistically assuming that the population is initialized in $A_1$. This removes the term ``$\sum_{i=1}^{m-1} \Prob{A_i} \cdot$'' and only leaves $\sum_{j=1}^{m-1} 1/s_j$. This way of simplifying upper bounds can be used for all results presented hereinafter.

The fitness-level method yields good upper bounds in many cases. This includes situations where an evolutionary algorithm typically moves through increasing fitness levels, without skipping too many levels~\cite{Sudholt2010a}. It only gives crude upper bounds in case values $s_i$ are dominated by search points from which the probability of leaving $A_i$ is much lower than for other search points in $A_i$ or if there are levels with difficult local optima (\ie, large values $1/s_i$) that are only reached with a small probability.

Using the expectation bounds from Section~\ref{sec:tail-bounds-and-expectations} we now show in Theorem \ref{the:upper-bound-for-A}: For both schemes, A and B, in the upper bound for the expected parallel time the expected sequential waiting time can be replaced by its logarithm. In addition, the expected sequential time is asymptotically not larger than the upper bound for the serial algorithm, derived by $f$-based partitions.

In the remainder of the paper we denote with $\paralleltime_{x}$ and $\sequentialtime_{x}, x\in \{\mathrm{A},\mathrm{B}\}$ the parallel time and the sequential time for the schemes~A and B, respectively.

\begin{theorem}
\label{the:upper-bound-for-A}
Given an $f$-based partition $A_1, \dots, A_m$,
\[
\E{\sequentialtime_{\mathrm{A}}} \le
2\sum_{i=1}^{m-1} \Prob{A_i} \cdot \sum_{j=i}^{m-1} \frac{1}{s_j}\;.
\]
If the partition is canonic then also
\[
\E{\paralleltime_{\mathrm{A}}} \le
2\sum_{i=1}^{m-1} \Prob{A_i} \cdot \sum_{j=i}^{m-1} \log\left(\frac{2}{s_j}\right) \;.
\]
\end{theorem}
The reason for the constant~2 in the $\log(2/s_j)$ term is to ensure that the term does not become smaller than~1; with a constant~1 the value $s_j = 1$ would even lead to a summand ${\log(1/s_j) = 0}$.
\begin{proof}
We only need to prove asymptotic bounds on the conditional expectations when starting in $A_i$, with a common constant hidden in all $O$-terms. The law of total expectation then implies the claim.

For Scheme~A we apply Lemma~\ref{lem:tail-bounds-and-expectations} with $k=0$.
This yields that the expected sequential time for leaving the current fitness level $A_j$ towards $A_{j+1} \cup \dots \cup A_m$ is
at most $2/s_j$ and the expected parallel time is at most $\log(1/s_j)+2 \le 2\log(2/s_j)$.
The expected sequential time is hence bounded by $2 \sum_{j=i}^{m-1} 1/s_j$ and the expected parallel time is at most $2\sum_{j=i}^{m-1} \log(2/s_j)$.
\end{proof}

We prove a similar upper bound for Scheme~B using arguments from the amortized analysis of algorithms~\cite[Chapter~17]{Cormen2001}. Amortized analysis is used to derive statements on the average running time of an operation or to estimate the total costs of a sequence of operations. It is especially useful if some operations may be far more costly than others and if expensive operations imply that many other operations will be cheap.
The basic idea of the so-called \emph{accounting method} is to let all operations pay for the costs of their execution. Operations are allowed to pay excess amounts of money to fictional accounts. Other operations can then tap this pool of money to pay for their costs. As long as no account becomes overdrawn, the total costs of all operations is bounded by the total amount of money that has been paid or deposited.

\begin{theorem}
\label{the:upper-bound-for-B}
Given an $f$-based partition $A_1, \dots, A_m$,
\[
\E{\sequentialtime_{\mathrm{B}}} \le
3\sum_{i=1}^{m-1} \Prob{A_i} \cdot \sum_{j=i}^{m-1} \frac{1}{s_j}\;.
\]
If the partition is canonic then also
\[
\E{\paralleltime_{\mathrm{B}}} \le
4\sum_{i=1}^{m-1} \Prob{A_i} \cdot \sum_{j=i}^{m-1} \log\left(\frac{2}{s_j}\right)\;.
\]
\end{theorem}
\begin{proof}
We use the accounting method to bound the expected sequential optimization time of~B as follows.
Assume the algorithm being on level~$j$ with a population size of~$2^k$. If the current generation passes without leaving the current fitness level, we pay $2^k$ to cover the costs for the sequential time in this generation. In addition, we pay another $2^k$ to a fictional bank account. In case the generation is successful in leaving $A_j$ and the previous generation was unsuccessful, we just pay $2^k$ and do not make a deposit.
In case the current generation is successful and the last unsuccessful generation was on fitness level~$j$, we withdraw $2^k$ from the bank account to pay for the current generation. In other words, the current generation is for free. This way, if there is a sequence of successful generations after an unsuccessful one on level~$j$ all but the first successful generations are for free.

Let us verify that the bank account cannot be overdrawn.
The basic argument is that, whenever the population size is decreased from, say, $2^{k+1}$ to $2^k$ then there must be a previous generation where the population size was increased from $2^k$ to $2^{k+1}$. It is easy to see that associating a decrease with the latest increase gives an injective mapping.
In simpler terms, the latest generation that has increased the population size from $2^k$ to $2^{k+1}$ has already paid for the current decrease to~$2^k$.

When in the upper bound for~A fitness level~$i$ takes sequential time $1 + 2 + \dots + 2^{k} = 2^{k+1}-1$ then for~B the total costs paid are $2(1+2+\dots+2^{k-1})+2^{k}$ as a successful generation does not make a deposit to the bank account. The total costs equal $2^{k+1}-2+2^k \le 3/2 \cdot (2^{k+1}-1)$.
In consequence, the total costs for Scheme~B are at most $3/2$ the costs for~A in A's upper bound. This proves the claimed upper bound for~B.

By the very same argument an upper bound for the expected parallel time for~B follows. Instead of paying $2^k$ and maybe making a deposit of $2^k$, we always pay 1 and always make a deposit of 1. When withdrawing money, we always withdraw 1. This proves that also $\E{\paralleltime_{\mathrm{B}}}$ is at most twice the corresponding upper bound for Scheme~A.
\end{proof}

The argument in the above proof can also be used for proving a general upper bound for the expected parallel optimization time for B.
When paying costs 2 for each fitness level, this pays for the successful generation with a population size of, say, $2^k$ and for one future generation where the population size might have to be doubled to reach $2^k$ again.

Imagine the sequence of population sizes over time and then delete all elements where the population size has decreased, including the associated generation where the population size was increased beforehand. In the remaining sequence the population size continually increases until, assuming a global optimum has not been found yet, after $n \log n$ generations a population size of at least $n^n$ is reached. In this case the probability of creating a global optimum by mutation is at least $(1-n^{-n})^{n^n} \approx 1/e$ as the probability of hitting any specific target point in one mutation is at least $n^{-n}$. The expected number of generations until this happens is clearly $O(1)$. We have thus shown the following.
\begin{corollary}
 For every function with $m$ function values we have $\E{\paralleltime_{\mathrm{B}}} \le 2m + n \log n + O(1)$.
\end{corollary}
This bound is asymptotically tight, for instance, for long path problems~\cite{Droste2002,Rudolph1997}. So, the $m$-term, in general, cannot be avoided.

When comparing A and B with respect to the expected parallel time, we expect B to perform better if the fitness levels have a similar degree of difficulty. This implies that there is a certain target level for the population size. Note, however, that such a target level does not exist in case the $s_i$-values are dissimilar.
In the case of similar $s_i$-values A might be forced to spend time doubling the population size for each fitness level until the target level has been reached. This waiting time is reflected by the $\log(2/s_j)$-terms in Theorem~\ref{the:upper-bound-for-A}. The following upper bound on~B shows that these log-terms can be avoided to some extent.
In the special yet rather common situation that improvements become harder with each fitness level, only the biggest such log-term is needed.
\begin{theorem}
\label{the:improved-upper-bound-for-B}
Given a canonical $f$-based partition $A_1, \dots, A_m$, $\E{\paralleltime_{\mathrm{B}}}$ is bounded by
\begin{align*}
\sum_{i=1}^{m-1} \Prob{A_i} \cdot & \Bigg(3(m-i-1) + \log\left(\frac{1}{s_i}\right) \\
&\qquad + \sum_{j=i+1}^{m-1} \positive{\log\left(\frac{1}{s_j}\right)-\log\left(\frac{1}{s_{j-1}}\right)}\Bigg)\;.
\end{align*}
If additionally $s_1 \ge s_2 \ge \dots \ge s_{m-1}$ then the bound simplifies to
\[
\sum_{i=1}^{m-1} \Prob{A_i} \cdot \left(3(m-i-1) + \log\left(\frac{1}{s_{m-1}}\right)\right)\;.
\]
\end{theorem}
\begin{proof}
The second claim immediately follows from the first one as the $\log$-terms form a telescoping sum.

For the first bound we again use arguments from amortized analysis.
By Lemma~\ref{lem:tail-bounds-and-expectations} if the current population size is~$2^k$ then the expected number of generations until an improvement from level~$i$ happens is at most ${\positive{\log(1/s_i)-k} + 2}$. This is a bound of 2 for $k \ge \log(1/s_i)$. We perform a so-called \emph{aggregate analysis} to estimate the total cost on all fitness levels. These costs are attributed to different sources. Summing up the costs for all sources will yield a bound on the total costs and hence on $\paralleltime_{\mathrm{B}}$.

In the first generation the fitness level~$i^*$ the algorithm starts on pays $\log(1/s_{i^*})$ to the global bank account. Afterwards costs are assigned as follows.
Consider a generation on fitness level~$i$ with a population size of~$2^k$.
\begin{itemize}
\item If the current generation is successful, we charge cost~2 to the fitness level; cost~1 pays for the effort in the generation and cost~1 is deposited on the bank account. In addition, each fitness level~$j$ that is skipped or reached during this improvement pays $\positive{\log(1/s_j)-\log(1/s_{j-1})}$ as a deposit on the bank account. Note that this amount is non-negative and it may be non-integer.
\item If $k \ge \log(1/s_i)$ and the current generation is unsuccessful we charge cost~1 to the fitness level.
\item If $k < \log(1/s_i)$ and the current generation is unsuccessful we withdraw cost~1 from our bank account.
\end{itemize}
By Lemma~\ref{lem:tail-bounds-and-expectations} the expected cost charged to fitness level~$i$ in unsuccessful generations (\ie, not counting the last successful generation) is at most~1.
Assuming for the moment that the bank account is never overdrawn, the overall expected cost for fitness level~$i$ is at most $1 + 2 + \positive{\log(1/s_j)-\log(1/s_{j-1})}$. Adding the costs for the initial fitness level yields the claimed bound.

We use the so-called \emph{potential method}~\cite[Chapter~17]{Cormen2001} to show that the bank account is never overdrawn. Our claim is that at any point of time there is enough money on the bank account to cover the costs of increasing the current population size to at least $2^{\log(1/s_j)}$ where $j$ is the current fitness level. We construct a potential function indicating the excess money on the bank account and show that the potential is always non-negative.

Let $\mu_t$ denote the population size in generation~$t$ and $\ell_t$ be the (random) fitness level in generation~$t$. By $b_t$ we denote the account balance on the bank account. We prove by induction that
\[
b_t \ge \positive{\log(1/s_{\ell_t})-\log(\mu_t)}\;.
\]
As this bound is always positive, this implies that the account is never overdrawn.
After the initial fitness level has made its deposit we have $b_1= \log(1/s_{\ell_1}) - 0$. Assume by induction that the bound holds for $b_t$.

If generation~$t$ is unsuccessful and $\log(\mu_t) \ge \log(1/s_{\ell_{t}})$ then the population size is doubled at no cost for the bank account. As by induction $b_t \ge 0$ we have $b_{t+1} = b_t \ge 0 = \positive{\log(1/(s_{\ell_t}))-\log(\mu_{t+1})}$.

If generation~$t$ is unsuccessful and $\log(\mu_t) < \log(1/s_{\ell_{t}})$ then the algorithm doubles its population size and withdraws 1 from the bank account. As $b_t$ is positive and $\log(\mu_{t+1}) = \log(\mu_t) + 1$, we have
\[
b_{t+1} = b_t - 1 = \log(1/s_{\ell_t}) - \log(\mu_{t}) - 1 = \log(1/s_{\ell_t}) - \log(\mu_{t+1}).
\]
If generation~$t$ is successful and the current fitness level increases from $i$ to some $j > i$, the account balance is increased by
\begin{align*}
& 1+\sum_{a=i+1}^{j} \positive{\log(1/s_a)-\log(1/s_{a-1})}\\
\ge\;& 1+\positive{\log(1/s_j)-\log(1/s_{i})}\;.
\end{align*}
This implies
\begin{align*}
b_{t+1} \ge\;& b_t + 1+\positive{\log(1/s_j)-\log(1/s_{i})}\\
\ge\;& \positive{\log(1/s_{i})-\log(\mu_t)} + 1 - \log(1/s_{i})-\log(\mu_{t+1})\\
\ge\;& \positive{\log(1/s_{j})-\log(\mu_t)} + 1\\
\ge\;& \positive{\log(1/s_{j})-\log(\mu_{t+1})}\;.\qedhere 
\end{align*}
\end{proof}

The upper bounds in this section can be easily adapted towards parallel EAs that do not perform migration and population size adaptation in every generation, but only every $\migint$ generations, for a migration interval $\migint \in \N$. Instead of considering the probability of leaving a fitness level in one generation, we simply consider the probability of leaving a fitness level in $\migint$ generations. This is done by considering $s_i' := 1-(1-s_i)^\migint$ instead of~$s_i$. The resulting time bounds, based on $s_1', \dots, s_{m-1}'$, are then with respect to the number of periods of $\migint$ generations. To get bounds on our original measures of time, we just multiply all bounds by a factor of~$\migint$.

\section{Lower Bounds for the Sequential Time}
\label{sec:lower-bounds}

In order to prove lower bounds for the expected sequential time we make use of recent results by Sudholt~\cite{Sudholt2010a}. He presented a new lower-bound method based on fitness-level arguments. If it is unlikely that many fitness levels are skipped when leaving the current fitness-level set then good lower bounds can be shown.

The lower bound applies to every algorithm $\algo$ in pseudo-Boolean optimization that only uses standard mutations (\ie, flipping each bit independently with probability $1/n$) to create new offspring. Such an EA is called a mutation-based EA. More precisely, every mutation-based EA $\algo$ works as follows. First, $\algo$ creates $\mu$ search points $x_1, \dots, x_\mu$ uniformly at random. Then it repeats the following loop. A counter $t$ counts the number of function evaluations; after initialization we have $t = \mu$. In one iteration of the loop the algorithm first selects one out of all search points $x_1, \dots, x_t$ that have been created so far. This decision is based on the fitness values $f(x_1), \dots, f(x_t)$ and, possibly, also the time index~$t$. It performs a standard mutation of this search point, creating an offspring $x_{t+1}$.

To make this work self-contained, we cite (a slightly simplified version of) the result here.
The performance measure considered is the number of function evaluations. This can be assumed to coincide with the number of offspring creations as every offspring needs to evaluated exactly once.
\begin{theorem}[\cite{Sudholt2010a}]
\label{the:lower-bound-method}
Consider a partition of the search space into non-empty sets
$A_1, \dots, A_m$ such that only $A_m$ contains global optima.
For a mutation-based EA $\algo$ we say that $\algo$ is in $A_i$ or on level $i$ if the best individual created so far is in~$A_i$. Let the probability of traversing from level $i$ to level $j$ in one mutation be at most $u_i \cdot \gamma_{i,j}$ and ${\sum_{j=i+1}^{m} \gamma_{i, j} = 1}$.
Assume that for all $j > i$ and some $0 < \chi \le 1$ it holds
$\gamma_{i, j} \ge \chi \sum_{k=j}^{m} \gamma_{i, k}$.
Then the expected number of function evaluations of $\algo$ on~$f$ is at least
\[
\chi \sum_{i=1}^{m-1} \Prob{A_i} \cdot \sum_{j=i}^{m-1} \frac{1}{u_j}\;.
\]
\end{theorem}
All population update schemes are compatible with this framework; every parallel mutation-based EA using an arbitrary population update scheme is still a mutation-based EA.
Offspring creations are performed in parallel in our algorithms, but one can imagine these operations to be performed sequentially. We can cast a parallel EA with parallel offspring creations as a sequential mutation-based EA that simulates the population management of an island model in the background. Recall that the selection in the notion of a mutation-based EA can be based on the time index $t$. Hence, a sequential mutation-based EA can keep track of the times when individuals on a specific island have been created or when individuals have immigrated from a different island. The algorithm can then simulate offspring creations for an island by allowing only individuals on the island to become parents.
There is one caveat: the parent selection mechanism in~\cite{Sudholt2010a} does not account for possibly randomized decisions made during migration. However, the proof of Theorem~\ref{the:lower-bound-method} goes through in case additional knowledge is used.

We introduce the notion of \emph{tight} fitness levels, where the success probabilities $s_i$ from the classical fitness-level method are exact up to a constant factor.
\begin{definition}
\label{def:tight-fitness-levels}
Call an $f$-based partition $A_1, \dots, A_m$ (asymptotically) tight for an algorithm $\mathcal{A}$ if there exist constants $c \ge 1 > \chi > 0$ and values $\gamma_{i,j}$ for $1 \le i,j \le m$ such that for each population in $A_i$ the following holds.
\begin{enumerate}
\item The probability of generating a population in $A_{i+1} \cup \dots \cup A_m$ in one mutation is at least $s_i$.
\item The probability of generating a population in $A_j$ in one mutation, $j > i$, is at most $c \cdot s_i \cdot \gamma_{i, j}$.
\item For the $\gamma_{i, j}$-values it holds that $\sum_{j=i+1}^{m} \gamma_{i, j} = 1$ and $\gamma_{i, j} \ge \chi \sum_{k=j}^{m} \gamma_{i, k}$ for all $i < j$.
\end{enumerate}
\end{definition}

Tight $f$-based partitions imply that the standard upper bound by $f$-based partitions~\cite{Wegener2002} is asymptotically tight. This holds for all elitist mutation-based algorithms, that is, mutation-based algorithms where the best fitness value in the population can never decrease.

\begin{theorem}
\label{the:lower-bound}
Consider an algorithm $\algo$ with an arbitrary population update strategy that only uses standard mutations for creating new offspring.
Given a tight $f$-based partition $A_1, \dots, A_m$ for a function $f$, we have
\[
\E{\sequentialtime} =
\mathord{\Omega}\mathord{\left(\sum_{i=1}^{m-1} \Prob{A_i} \cdot \sum_{j=i}^{m-1} \frac{1}{s_j}\right)}\;.
\]
\end{theorem}
\begin{proof}
The lower bound on~$\E{\sequentialtime}$ follows by a direct application of Theorem~\ref{the:lower-bound-method}. We already discussed that this theorem applies to all algorithms considered in this work. Setting $u_j := c s_j$ for all $1 \le j \le m$, $c$ and $\chi$ being as in Definition~\ref{def:tight-fitness-levels}, Theorem~\ref{the:lower-bound-method} implies
\[
\E{\sequentialtime} \ge \frac{\chi}{c}  \sum_{i=1}^{m-1} \Prob{A_i} \cdot \sum_{j=i}^{m-1} \frac{1}{s_j}\;.
\]
As both, $\chi$ and $c$, are constants, this implies the claim.
\end{proof}

This lower bound shows that for tight $f$-based partitions both our population update schemes produce asymptotically optimal results in terms of the expected sequential optimization time, assuming no cost of communications.

\section{Non-oblivious Update Schemes}
\label{sec:non-oblivious}

We also briefly discuss update schemes that are tailored towards particular functions, in order to judge the performance of our oblivious update schemes.

Non-oblivious population update schemes may allow for smaller upper bounds for the expected parallel time than the ones seen so far. When the population update scheme has complete knowledge on the function~$f$ and the $f$-based partition, an upper bound can be shown where each fitness level contributes only a constant to the expected parallel time.
By $\sequentialtime_{\mathrm{no}}$ and $\paralleltime_{\mathrm{no}}$ we denote the sequential and parallel times of the considered non-oblivious scheme.
\begin{theorem}
\label{the:upper-bound-for-C}
Given an arbitrary $f$-based partition $A_1, \dots, A_m$, there is a tailored population update scheme for which
\[
\E{\sequentialtime_{\mathrm{no}}} = \mathord{O}\mathord{\left(
\sum_{i=1}^{m-1} \left(\Prob{A_i} \cdot \sum_{j=i}^{m-1} \frac{1}{s_j}\right)\right)}
\]
and
\[
\E{\paralleltime_{\mathrm{no}}} = \mathord{O}\mathord{\left(
\sum_{i=1}^{m-1} \Prob{A_i} \cdot (m-i-1)\right)}\;.
\]
In particular, $\E{\paralleltime_{\mathrm{no}}} = O(m)$.
\end{theorem}
\begin{proof}
The update scheme chooses to use $\lceil 1/s_i \rceil$ islands if the algorithm is in $A_i$.
Then the probability of finding an improvement in one generation is at least $1-(1-s_i)^{1/s_i} \ge 1-1/e$. The expected parallel time until this happens is at most $e/(e-1)$ and so the expected sequential time is at most $e/(e-1) \cdot \ceil{1/s_i} \le 2e/(e-1) \cdot 1/s_i$. Summing up these expectations for all fitness levels from $i$ to~$m-1$ proves the two bounds.
\end{proof}

In some situations it is possible to design schemes that perform even better than the above bound suggests. For instance, for trap functions the best strategy would be to use a very large population in the first generation so that the optimum is found with high probability, and before the algorithm is tricked to increasing the distance to the global optimum.

\section{Bounds for Example Functions}
\label{sec:example}

The previous bounds are applicable in a very general context, with arbitrary fitness functions. We also give results for selected example functions to estimate possible speed-ups in more concrete settings.

We consider the same example functions and function classes that have been investigated in~\cite{Lassig2010a}. The goal is the maximization of a pseudo-Boolean function $f \colon \{0, 1\}^n \to \R$. For a search point $x \in \{0, 1\}^n$ write $x = x_1 \dots x_n$, then $\OM(x) := \sum_{i=1}^n x_i$ counts the number of ones in $x$ and $\LO(x) := \sum_{i=1}^n \prod_{j=1}^i x_i$ counts the number of leading ones in $x$. A function is called \emph{unimodal} if every non-optimal search point has a Hamming neighbor (\ie, a point with Hamming distance 1 to it) with strictly larger fitness.
For $1 \le k \le n$ we also consider
\[
\Jump_k := \begin{cases}
k + \sum_{i=1}^n x_i, & \textrm{if $\sum_{i=1}^n x_i \le n-k$ or $x = 1^n$ ,}\\
\sum_{i=1}^n (1-x_i) & \textrm{otherwise}\;.
\end{cases}
\]
This function has been introduced by Droste, Jansen, and Wegener~\cite{Droste2002} as a function with tunable difficulty. Evolutionary algorithms typically have to perform a jump to overcome a gap by flipping $k$ specific bits.

For these functions we obtain bounds for $\sequentialtime$ and $\paralleltime$ as summarized in Table \ref{tab:running-times}.
The lower bounds for $\E{\sequentialtime}$ on \OM and \LO follow directly from \cite{Sudholt2010a} for all schemes.

\begin{table}[h!]
\centering\small
\begin{tabular}{|p{2.5cm}|l|p{1.9cm}|p{2.3cm}|}\hline
          & Scheme & $\E{\sequentialtime}$ & $\E{\paralleltime}$\\\hline
$\OM$     & A         & $\Theta(n\log n)$       & $O(n\log n)$    \\
          & B         & $\Theta(n\log n)$       & $O(n)$          \\
          & non-oblivious & $\Theta(n\log n)$       & $O(n)$          \\\hline
$\LO$     & A         & $\Theta(n^2)$           & $\Theta(n\log n)$    \\
          & B         & $\Theta(n^2)$           & $O(n)$          \\
          & non-oblivious & $\Theta(n^2)$           & $O(n)$          \\\hline
unimodal $f$ & A         & $O(dn)$                 & $O(d\log n)$\\
with $d$ $f$-values          & B         & $O(dn)$                 & $O(d+\log n)$\\
          & non-oblivious & $O(dn)$                 & $O(d)$\\\hline
$\Jump_k$ & A         & $O(n^k)$           & $O(n\log n)$  \\
with $k\geq 2$ & B         & $O(n^k)$           & $O(n+k\log n)$  \\
          & non-oblivious & $O(n^k)$           & $O(n)$          \\\hline
\end{tabular}\smallskip
\caption{Asymptotic bounds for expected parallel running times $\E{\paralleltime}$ and expected sequential running times $\E{\sequentialtime}$ for the parallel \EA and the \EAL with adaptive population models.}
\label{tab:running-times}
\end{table}

\begin{theorem}\label{the:upperBoundsParallelTime}
For the parallel \EA and the \EAL with adaptive population models the upper bounds for $E(\sequentialtime)$ and $E(\paralleltime)$ hold as given in Table \ref{tab:running-times}.
\end{theorem}
\begin{proof}
The upper bounds for Scheme~A follow from Theorem~\ref{the:upper-bound-for-A}, for Scheme~B from Theorems~\ref{the:upper-bound-for-B} and~\ref{the:improved-upper-bound-for-B} and for the non-oblivious scheme from Theorem~\ref{the:upper-bound-for-C}. 
Starting pessimistically from the first fitness level, the following bounds hold:
\begin{itemize}
\item For \OM we use the canonical $f$-based partition $A_i := \{x \mid \OM(x) = i\}$ and the corresponding success probabilities $s_i \ge (n-i)/n \cdot (1-1/n)^{n-1} \ge (n-i)/(en)$. Hence, $E(\paralleltime_{\mathrm{A}}) \leq 2 \sum_{i=1}^{n-1} \log(\frac{2en}{n-i})\leq 2n\log(2en)=O(n\log n)$,
    \begin{eqnarray*}
E(\sequentialtime_{\mathrm{A}})&\le& 2\sum_{i=0}^{n-1} \frac{1}{s_i}\leq 2\sum_{i=0}^{n-1} \frac{en}{n-i}\\
                  &=& 2en\sum_{i=1}^n \frac{1}{i}=2en\cdot[(\ln n)+1]\;,
\end{eqnarray*}
$E(\paralleltime_{\mathrm{B}}) \leq (3(n-2)+\log(2en))=O(n)$ and $E(\sequentialtime_{\mathrm{B}})\le 3en\cdot[(\ln n)+1]$, $E(\paralleltime_{\mathrm{no}})=O(n)$ and $E(\sequentialtime_{\mathrm{no}})=O(n\log n)$.
\item For \LO we use the canonical $f$-based partition $A_i := \{x \mid \LO(x) = i\}$ and the corresponding success probabilities $s_i \ge 1/n \cdot (1-1/n)^{n-1} \ge 1/(en)$. Hence, $E(\paralleltime_{\mathrm{A}}) \leq 2 \sum_{i=0}^{n-1} \log(2en)=2n\log(2en)=O(n\log n)$,
    \[
E(\sequentialtime_{\mathrm{A}})\le 2\sum_{i=0}^{n-1} \frac{1}{s_i}\leq 2\sum_{i=0}^{n-1} en =2en^2\;,
\]
$E(\paralleltime_{\mathrm{B}}) \leq (3(n-2)+\log(en))=O(n)$, $E(\sequentialtime_{\mathrm{B}})\le 3en^2$, $E(\paralleltime_{\mathrm{no}})=O(n)$ and $E(\sequentialtime_{\mathrm{no}})=O(n^2)$.
\item For unimodal functions with $d$ function values 
we use corresponding success probabilities $s_i \ge 1/(en)$. Hence,
    $E(\paralleltime_{\mathrm{A}}) \leq 2 \sum_{i=1}^{d-1} \log (2en)\leq 2d\log(2en)=O(dn)$,
    \[
E(\sequentialtime_{\mathrm{A}})\le 2\sum_{i=1}^{d-1} \frac{1}{s_i}\leq 2\sum_{i=1}^{d-1} en =2edn\;,
\]
$E(\paralleltime_{\mathrm{B}}) \leq 3(d-2)+\log(en)=O(d+\log n)$, $E(\sequentialtime_{\mathrm{B}})=3edn$, $E(\paralleltime_{\mathrm{no}})=O(d)$ and $E(\sequentialtime_{\mathrm{no}})=O(dn)$.
\item For $\Jump_k$ functions with $k \ge 2$ and all individuals having neither $n-k$ nor $n$ 1-bits, an improvement is found by either increasing or decreasing the number of 1-bits. This corresponds to optimizing \ONEMAX. In order to improve a solution with $n-k$ 1-bits, a specific bit string with Hamming distance $k$ has to be created, which has probability $s_{n-k}$ at least
\[
 \left(\frac{1}{n}\right)^k \cdot \left(1-\frac{1}{n}\right)^{n-k}\geq \left(\frac{1}{n}\right)^k \cdot \left(1-\frac{1}{n}\right)^{n-1}\geq\frac{1}{en^k}\;.
\]
Hence, $E(\paralleltime_{\mathrm{A}}) \leq O(n \log n) + 2 \log(en^k) \le O(n \log n) + 2k\log(en)=O(n\log n)$,
$
E(\sequentialtime_{\mathrm{A}})\le O(n^k),
$
$E(\paralleltime_{\mathrm{B}}) \leq O(n)+k\log(en)=O(n+k\log n)$, $E(\sequentialtime_{\mathrm{B}})\le O(n^k)$, $E(\paralleltime_{\mathrm{no}}) = O(n)$ and $E(\sequentialtime_{\mathrm{no}})=O(n^k)$.\qedhere
\end{itemize}
\end{proof}

It can be seen from Table~\ref{tab:running-times} that both our schemes lead to significant speed-ups in terms of the parallel time.
The speed-ups increase with the difficulty of the function. This becomes obvious when comparing the results on \OM and \LO and it is even more visible for $\Jump_k$.

The upper bounds for $\E{\paralleltime_{\mathrm{B}}}$ are always asymptotically lower than those for $\E{\paralleltime_{\mathrm{A}}}$, except for $\Jump_k$ with $k = \Theta(n)$. However, without corresponding lower bounds we cannot say whether this is due to differences in the real running times or whether we simply proved tighter guarantees for B.
We therefore consider the function \LO in more detail and prove a lower bound for A. This demonstrates that Scheme~B can be asymptotically better than Scheme~A on a concrete problem.
\begin{theorem}
\label{the:parallel-time-on-LO}
For the parallel \EA and the \EAL with adaptive population models on \LO we have $\E{\paralleltime_{\mathrm{A}}} = \Omega(n \log n)$.
\end{theorem}
\begin{proof}
We consider a pessimistic setting (pessimistic for proving a lower bound) where an improvement has probability exactly~$1/n$. This ignores that all leading ones have to be conserved in order to increase the best \LO-value. We show that with probability $\Omega(1)$ at least $n/30$ improvements are needed in this setting. As by Lemma~\ref{lem:tail-bounds-and-expectations} the expected waiting time for an improvement is at least $\max\{0, (\log n)-3\}$, the conditional expected parallel time is $\Omega(n \log n)$. By the law of total expectation, also the unconditional expected parallel time is then $\Omega(n \log n)$.

Let us bound the expected increase in the number of leading ones on one fitness level. Let $\paralleltime_i$ denote the random number of generations until the best fitness increases when the algorithm is on fitness level~$i$.
By the law of total expectation the expected increase in the best fitness in this generation equals
\begin{equation}
\label{eq:LO-increase}
\sum_{t=1}^\infty \Prob{\paralleltime_i = t} \cdot \E{\text{\LO-increase} \mid \paralleltime_i = t}\;.
\end{equation}
The expected increase in the number of leading ones can be estimated as follows. With $\paralleltime_i = t$ the number of mutations in the successful generation is $2^{t-1}$. Let $I$ denote the number of mutations that increase the current best \LO-value. A well-known property of \LO is that when the current best fitness is~$i$ then the bits at positions $i+2, \dots, n$ are uniform.
Bits that form part of the leading ones after an improvement are called \emph{free riders}. The probability of having $k$ free riders is thus $2^{-k}$ (unless the end of the bit string is reached) and the expected number of free riders is at most $\sum_{k=0}^{\infty} 2^{-k} = 1$.

The uniformity of ``random'' bits at positions $i+2, \dots, n$ holds after any specific number of mutations and in particular after the mutations in generation $\paralleltime_i$ have been performed.
However, when looking at multiple improvements, the free-rider events are not necessarily independent as the ``random'' bits are very likely to be correlated. The following reasoning avoids these possible dependencies. We consider the improvements in generation $\paralleltime_i$ one-by-one. If $F_1$ denotes the random number of free riders gained in the first improvement, when considering the second improvement the bits at positions $i+3+F_1, \dots, n$ are still uniform. In some sense, we give away the free riders from a fitness improvements for free for all following improvements.
This leads to an estimation of $1+F_1$ for the gain in the number of leading ones.

Iterating this argument, the expected total number of leading ones gained is thus bounded by $2I$, the expectation being taken for the randomness of free riders. Also considering the expectation for the random number of improvements yields the bound $2 \E{I \mid I \ge 1}$ as $I$ has been defined with respect to the last (\ie successful) generation. We also observe $\E{I \mid I \ge 1} \le 1 + \E{I} \le 1+2^t/n$. Plugging this into Equation~\eqref{eq:LO-increase} yields
\allowdisplaybreaks[3]
\begin{align*}
&\sum_{t=1}^\infty \Prob{\paralleltime_i = t} \cdot (2+2^{t+1}/n)\\
=\;& 2 + 2\sum_{t=0}^\infty \Prob{\paralleltime_i = t+1} \cdot 2^{t+1}/n\\
\le\;& 2 + 2\sum_{t=0}^\infty \Prob{\paralleltime_i > t} \cdot 2^{t+1}/n\\
\le\;& 2 + 2\sum_{t=0}^{\ceil{\log n}} 2^{t+1}/n + 2\sum_{t=\ceil{\log n}+1}^\infty \Prob{\paralleltime_i > t} \cdot 2^{t+1}/n\;.
\end{align*}
\allowdisplaybreaks[0]
The first sum is at most~16. Using Lemma~\ref{lem:tail-bounds-and-expectations} to estimate the second sum, we arrive at the lower bound
\begin{align*}
&18 + 2\sum_{\alpha=0}^\infty \Prob{\paralleltime_i > \ceil{\log n} + \alpha + 1} \cdot 2^{\ceil{\log n} + \alpha+2}/n\\
\le\;& 18 + 2\sum_{\alpha=0}^\infty \exp(2^{-\alpha}) \cdot 2^{\ceil{\log n} + \alpha+2}/n\\
\le\;& 18 + 16 \cdot \sum_{\alpha=0}^\infty \exp(2^{-\alpha}) \cdot 2^{\alpha}\\
<\;& 29.8\;.
\end{align*}
With probability $1/2$ the algorithm starts with no leading ones, independently from all following events. The expected number of leading ones after $n/30$ improvements is at most $29.8/30 \cdot n$. By Markov's inequality the probability of having created $n$ leading ones is thus at most $29.8/30$ and so with probability $1/2 \cdot 0.2/30 = \Omega(1)$ having $n/30$ improvements is not enough to find a global optimum.
\end{proof}

\section{Generalizations \& Extensions}
\label{sec:extensions}

We finally discuss generalizations and extensions of our results.

One interesting question is in how far our results change if the population is not doubled or halved, but instead multiplied or divided by some other value $b > 1$. Then the results would change as follows. With some potential adjustments to constant factors, the $\log$-terms in the parallel optimization times in Theorems~\ref{the:upper-bound-for-A}, \ref{the:upper-bound-for-B} and~\ref{the:improved-upper-bound-for-B} would have to be replaced by $\log_b$. For the sequential optimization times stated in these theorems one would need to multiply these bounds by $b/2$. This means that a larger $b$ would further decrease the parallel optimization times at the expense of a larger sequential optimization time.

Our analyses can also be transferred towards the adaptive scheme presented by Jansen, De~Jong, and Wegener~\cite{Jansen2005a}. Recall that in their scheme the population size is divided by the number of successes. In case of one success the population size remains unchanged. This only affects the constant factors in our upper bounds. When the number of successes is large, the population size might decrease quickly.
In most cases, however, the number of successes will be rather small; for instance, the lower bound for \LO, Theorem~\ref{the:parallel-time-on-LO}, has shown that the expected number of successes in a successful generation is constant. However, it might be possible that after a difficult fitness level an easier fitness level is reached and then the number of successes might be much higher. In an extreme case their scheme can decrease the population size like Scheme~A. In some sense, their scheme is somewhat ``in between'' A and B. With a slight adaptation of the constants, the upper bound for Scheme~A from Theorem~\ref{the:upper-bound-for-A} can be transferred to their scheme.

Another extension of the results above is towards maximum population sizes.
Although we have argued in Section~\ref{sec:tail-bounds-and-expectations} that the population size does not blow up too much, in practice the maximum number of processors might be limited. The following theorem about $E(\paralleltime_{\mathrm{A}})$ for maximum population sizes can be proven by applying arguments from~\cite{Lassig2010a}.

\begin{theorem}
\label{the:method-max}
The expected parallel optimization time of Scheme~A for a maximum population size $\mu:=\mu_{\max} > 1$ is bounded by
\[
E(\paralleltime_{\mathrm{A}}) \;\le\;  m\cdot [\log \mu_{\max}+2] +\frac{2}{\mu_{\max}} \sum_{i=1}^{m-1} \frac{1}{s_i}\;.
\]
\end{theorem}
\begin{proof}
We pessimistically estimate the expected parallel time by the time until the population consists of $\mu_{\max}$ islands plus the expected optimization time if $\mu_{\max}$ islands are available. The time until $\mu_{\max}$ islands are involved is $\log \mu_{\max}$ on one fitness level. Hence, summing up all levels pessimistically gives $m\log \mu_{\max}$. For $\mu_{\max}$ islands the success probability on fitness level $i$ with success probability $s_i$ for one island is given by $1-(1-s_i)^{\mu_{\max}}$. Hence, the expected time for leaving fitness level~$i$ if $\mu_{\max}$ islands are available is at most $1/[1-(1-s_i)^{\mu_{\max}}]$.
Now we consider two cases.

If $s_i\cdot \mu_{\max}\leq 1$ we have
$1-(1-s_i)^{\mu_{\max}} \geq 1-(1-s_i\mu_{\max}/2)=s_i\mu_{\max}/2$ because for all $0\leq xy\leq 1$ it holds $(1-x)^y\leq 1-xy/2$~\cite[Lemma~1]{Lassig2010a}.
Otherwise, if $s_i\cdot \mu_{\max} >1$ we have
$
1-(1-s_i)^{\mu_{\max}} \geq 1-e^{-s_i \mu_{\max}} \geq 1-\frac{1}{e}
$.
Thus,
\begin{eqnarray*}
\sum_{i=1}^{m-1} \frac{1}{1-(1-s_i)^{\mu_{\max}}}&\leq& \sum_{i=1}^{m-1} \max\left\{\frac{1}{1-1/e},\frac{2}{\mu_{\max} \cdot s_i}\right\}\\
&\leq& m\cdot \frac{e}{e-1}+\frac{2}{\mu_{\max}} \sum_{i=1}^{m-1} \frac{1}{s_i}\;.
\end{eqnarray*}
Adding the expected waiting times until $\mu_{\max}$ islands are involved yields the claimed bound.
\end{proof}

In terms of our test functions \OM, LO, unimodal functions, and $\Jump_k$, this leads to the following result that can be proven like Theorem~\ref{the:upperBoundsParallelTime}.
\begin{corollary}
\label{cor:application-mumax}
For the parallel \EA and the \EAL with Scheme~A the following holds for a maximum population size $\mu:=\mu_{\max} > 1$:
\begin{itemize}
\item $E(\paralleltime_{\mathrm{A}})=O(n\log \mu_{\max}+n\log(n)/ \mu_{\max})$ for $\OM$, which gives $O(n\log\log n)$ for $\mu_{\max}=\log n$,
\item $E(\paralleltime_{\mathrm{A}})=O(n\log \mu_{\max}+n^2/ \mu_{\max})$ for $\LO$, which gives $O(n\log n)$ for $\mu_{\max}=n$,
\item $E(\paralleltime_{\mathrm{A}})=O(d\log \mu_{\max}+dn/ \mu_{\max})$ for unimodal functions with $d$ function values, which gives $O(d\log n)$ for $\mu_{\max}=n$,
\item $E(\paralleltime_{\mathrm{A}})=O(n\log \mu_{\max}+n^k/ \mu_{\max})$ for $\Jump_k$, which gives $O(nk\log n)$ for $\mu_{\max}=n^{k-1}$.
\end{itemize}
\end{corollary}
Note that Corollary~\ref{cor:application-mumax} has led to an improvement of $\E{\paralleltime_{\mathrm{A}}}$ from $O(n \log n)$ to $O(n \log \log n)$ for $\mu_{\max} = \log n$. This obviously also holds in the setting of unrestricted population sizes.

\section{Conclusions}
\label{sec:Conclusions}

We have presented two schemes for adapting the offspring population size in evolutionary algorithms and, more generally, the number of islands in parallel evolutionary algorithms. Both schemes double the population size in each generation that does not yield an improvement. Despite the exponential growth, the expected sequential optimization time is asymptotically optimal for tight $f$-based partitions. In general, we obtain bounds that are asymptotically equal to upper bounds via the fitness-level method.

In terms of the parallel computation time expected waiting times on a fitness level can be replaced by their logarithms for both schemes, compared to a serial EA.
This yields a tremendous speed-up, in particular for functions where finding improvements is difficult.
Scheme~B, doubling or halving the population size in each generation, turned out to be more effective than resets to a single island as in Scheme~A.
This is because B can quickly decrease the population size if necessary. The effort spent while this happens does not affect the asymptotic bounds for expected parallel and sequential times.

Apart from our main results, we have introduced the notion of tight $f$-based partitions and new arguments from amortized analysis of algorithms to the theory of evolutionary algorithms.

An open question is
how our schemes perform in situations where the fitness-level method does not provide good upper bounds. In this case our bounds may be off from the real expected running times. In particular, there may be examples where
increasing the offspring population size by too much might be detrimental. One constructed function where large offspring populations perform badly was presented in~\cite{Jansen2005a}. Future work could characterize function classes for which our schemes are efficient in comparison to the real expected running times. The notion of tight $f$-based partitions is a first step in this direction.

\subsection*{Acknowledgments}
The authors would like to thank the German Academic Exchange Service for funding their research. Part of this work was done while both authors were visiting the International Computer Science Institute in Berkeley, CA, USA.
The second author was partially supported by EPSRC grant EP/D052785/1.
The authors thank Carola Winzen for many useful suggestions that helped to improve the presentation.

\balance
\bibliographystyle{abbrv}
\newcommand{\noopsort}[1]{} \newcommand{\printfirst}[2]{#1}
  \newcommand{\singleletter}[1]{#1} \newcommand{\switchargs}[2]{#2#1}

\end{document}